\documentclass[conference]{IEEEtran}
\IEEEoverridecommandlockouts
\usepackage{cite}
\usepackage{overpic}
\usepackage{amsmath,amssymb,amsfonts}
\usepackage{graphicx}
\usepackage{textcomp}
\usepackage{xcolor}
\usepackage{float}
\usepackage{amsthm}
\usepackage{graphicx}
\usepackage{epstopdf}
\usepackage{subfigure}
\usepackage{adjustbox}
\usepackage{stfloats}
\usepackage{amsmath,bm}
\usepackage{amsfonts}
\usepackage{amssymb}
\usepackage{color}
\usepackage{multirow}
\usepackage{multicol}
\usepackage{soul,xcolor}
\usepackage{algorithm}
\usepackage{algorithmic}
\usepackage{url}

\theoremstyle{plain}

\newtheorem{lemma}{Lemma}

\newcommand{\vect}[1]{\mathbf{#1}}

\def\Htran{\mbox{\tiny $\mathrm{H}$}}
\def\Ttran{\mbox{\tiny $\mathrm{T}$}}

\begin{document}

\title{Joint Access and Fronthaul Resource Allocation for Cell-Free Massive MIMO with Wireless Fronthaul }

\author{ \IEEEauthorblockN{Ozan Alp Topal\IEEEauthorrefmark{1}, Özlem Tuğfe Demir\IEEEauthorrefmark{2}, Emil Björnson\IEEEauthorrefmark{1}, and Cicek Cavdar\IEEEauthorrefmark{1}}
\IEEEauthorblockA{ \IEEEauthorrefmark{1}{Department of Communication Systems, KTH Royal Institute of Technology, Stockholm, Sweden
		}   \\
        \IEEEauthorrefmark{2}{Department of Electrical and Electronics Engineering, Bilkent University, Ankara, Türkiye
		} \\
		\IEEEauthorblockA{E-mail: 
        \{oatopal, emilbjo, cavdar\}@kth.se,  ozlemtugfedemir@bilkent.edu.tr}
}

\thanks{ This work has been part of the Celtic-Next project RAI-6Green: Robust and AI Native 6G for Green Networks with project-id: C2023/1-9, partly supported by Swedish funding agency Vinnova. The work by \"O. T. Demir was supported by 2232-B International Fellowship for Early Stage Researchers Programme funded by the Scientific and Technological Research Council of T\"urkiye. E.~Bj\"ornson was supported by the IS24-0190 grant from SSF.}

}

\maketitle

\begin{abstract}
Wireless fronthaul is a key enabler of flexible and scalable cell-free massive MIMO systems, but its limited capacity poses significant challenges for maintaining high and uniform user performance.  In this work, we analyze the performance of a cell-free massive MIMO network with wireless fronthaul under realistic low physical layer functional splits. We propose a joint access and fronthaul resource allocation algorithm that maximizes the minimum user equipment (UE) spectral efficiency while satisfying fronthaul load constraints. Our analysis reveals that power allocation over the wireless fronthaul follows a modified water-filling structure, where the water level is jointly determined by the access and fronthaul channel gains. Furthermore, we show that severe fronthaul limitations not only reduce UE rates but also introduce spatial performance disparities depending on the cloud location. Finally, we demonstrate that split option 8 is impractical under wireless fronthaul constraints, underscoring the importance of dynamic fronthaul bit allocation to reduce fronthaul load and enable efficient system operation.
\end{abstract}

\section{Introduction}

Cell-free massive MIMO (multiple-input multiple-output) has emerged as a promising candidate for future mobile networks, thanks to its ability to deliver uniform service quality across user equipment (UEs). 
The performance gains of cell-free massive MIMO systems rely on the ability of distributed access points (APs) to perform coherent joint transmission and reception \cite{cfmMIMOOr}. In practical deployments with a large number of low-cost APs, this coordination is typically enabled by centralized processing at a cloud unit (or multiple cloud units), where baseband signals are jointly processed. Achieving such tight cooperation requires precise synchronization among APs, which in turn necessitates that key physical layer (PHY) functions be implemented centrally \cite{larsen_survey_2019}. Consequently, cell-free systems must adopt low-PHY functional splits, which require high-capacity fronthaul links \cite{deployment_IDS}.

In earlier studies, the fronthaul network is assumed to be based on optical fiber, providing high reliability and bandwidth \cite{cell-free-book}. However, fiber deployment also imposes significant infrastructure costs and reduces radio deployment flexibility. In contrast, wireless fronthaul has emerged as a promising alternative, offering a more flexible and cost-effective solution, but it limits the fronthaul capacity and coverage \cite{umurhan, wireless_f}. These prior works rely on simplified fronthaul assumptions, often relating fronthaul rate requirements to user data rates. In practice, supporting coherent joint processing requires low-PHY functional splits (split option $8$, $7.1$ and $7.2$ as detailed in \cite{energyjournal}), which impose $10$ times higher fronthaul rate than the assumed, and not only scale with the spectral bandwidth but also scale with the number of antennas at the APs (specifically in split options $8$ and $7.1$) \cite{icton}. This makes the fronthaul a key performance bottleneck, particularly in wireless implementations with limited capacity. The authors of \cite{AsilomarConf24,energyjournal, neetu} consider low-PHY splits in cell-free massive MIMO, but they focus on minimizing total power consumption, overlooking the fundamental performance bottlenecks imposed by wireless fronthaul capacity.

To address this gap, in this work, we develop a joint access and fronthaul resource allocation framework to provide uniform service quality in a downlink cell-free massive MIMO system with wireless fronthaul. The proposed algorithm maximizes the minimum spectral efficiency among the UEs through dynamically controlling the active number of AP antennas in the access and the power allocation in the fronthaul. For the simplified single-UE case, the optimal power allocation over the fronthaul follows a modified water-filling structure, where the allocation depends jointly on the access and fronthaul channel conditions. In numerical analysis, we demonstrate antenna activation patterns and minimum UE rate performance under wireless fronthaul, considering functional split options 8 and 7.1.  
        \vspace{-2mm}

\section{System Model}
\label{sec:system_model}
        \vspace{-1mm}

We consider the downlink of a cell-free massive MIMO system operating in time-division duplex (TDD) mode. The system consists of $L$ APs connected to the cloud via a high-frequency (mmWave) wireless fronthaul link, where APs and the cloud have line-of-sight (LOS) connectivity.

The APs serve $K$ single-antenna UEs over a mid-band frequency (sub-6 GHz) channel. The channel between APs and the cloud unit will be referred to as the fronthaul channel, while the channel between APs and UEs will be referred to as the access channel. For the access links, we assume uncorrelated Rayleigh fading channels as in \cite{interdonato2020local}, i.e., $\vect{h}_{l,k} \sim \mathcal{N}_{\mathbb{C}}( \boldsymbol{0}, \beta_{l,k} \vect{I}_{M^{\mathrm{ac}}})$.   Each AP is equipped with $M^\mathrm{ac}$ antennas for the access channel and $M^\mathrm{frh}$ antennas for the fronthaul channel. In the proposed system, we consider that each AP can activate or deactivate its access antennas based on the quality-of-service (QoS) requirements and fronthaul load limitations. Therefore, we denote the activated antennas at AP $l$  by $M_l \in \{0,\ldots,M^\mathrm{ac}\}$, where $M_l=0$ means that AP $l$ is deactivated. The bandwidth utilized in the access channel and the fronthaul channel are denoted as $B^{\mathrm{ac}}$ and $B^{\mathrm{frh}}$, respectively. We let $\tau_c$ denote the number of symbols in a TDD frame, where it consists of uplink training (with $\tau_p$ symbols for the training) and downlink data transmission (with $\tau_d=\tau_c-\tau_p$ symbols). Due to space limitations, we omit the explanations for the uplink training phase, but the implementation is detailed in  \cite[Chapter 5]{cell-free-book}. 
        \vspace{-1mm}
\subsection{Downlink Data Transmission in Access Channel}
        \vspace{-1mm}

After receiving the downlink message signals from the cloud, APs can simultaneously transmit the same message signal to enable a coherently enhanced signal at each receiving UE. The transmit signal from AP $l$ is given as 
     \vspace{-1mm}
\begin{equation}
    \vect{x}_l = \sum_{k=1}^K \sqrt{\rho_{l,k}} \vect{w}_{l,k} \varsigma_{k},      \vspace{-1mm}
\end{equation}
where $\rho_{l,k}$ is the  transmit power assigned for UE $k$ at AP $l$ and $\sum_{k=1}^{K}\rho_{l,k} \leq P_t, \, \forall l$. $\vect{w}_{l,k} \in \mathbb{C}^{M_l}$ is the precoding vector used by AP $l$ towards UE $k$ with $\mathbb{E}\left[\|\vect{w}_{l,k} \|^2\right] = 1$, where  $\|\vect{w}\|$ denotes Euclidean norm of a vector $\vect{w}$ and $\mathbb{E}\left[ \cdot \right]$ is the expectation operator. $\varsigma_{k}$ is the unit-power downlink data symbol of UE $k$,  $\mathbb{E}\left[|\varsigma_{k}|^2\right]=1$. We neglected the pilot contamination effect to focus on wireless fronthaul limitations.

We consider local protective partial zero-forcing (PPZF) as the downlink precoding scheme \cite{interdonato2020local}. 
PPZF can provide a better balance between array gain and interference cancellation than other distributed precoding schemes. In PPZF, each AP divides UEs into two distinct sets: strong-channel UEs, $\mathcal{S}_l$, and weak-channel UEs, $\mathcal{W}_l$, where $\mathcal{S}_l \bigcap \mathcal{W}_l = \varnothing$. Then, APs utilize zero forcing (ZF) precoding for the strong-channel UEs, and protective maximum ratio transmission (MRT) for the weak-channel UEs. 
The exact expressions of precoding vectors are also omitted due to space limitations, but can be found in \cite{interdonato2020local}. An achievable ergodic spectral efficiency (SE) for UE $k$ is  $\mathrm{SE}_k = \frac{\tau_c - \tau_p}{\tau_c} \log_2(1+ \mathrm{SINR}_k)$, where $\mathrm{SINR}_k$ is named as the effective signal-to-interference-plus-noise ratio (SINR) of UE $k$, and for the considered precoding scheme is given as
\vspace{-1mm}
\begin{equation}
    \operatorname{SINR}_k =\frac{\left(\sum_{l=1}^L \sqrt{\left(M_l-\tau_{\mathcal{S}_l}\right) \rho_{l, k} \gamma_{l, k}}\right)^2}{\sum_{t=1}^K \sum_{l=1}^L \rho_{l, t}\left(\beta_{l, k}-\delta_{l, k} \gamma_{l, k}\right)+\sigma^2},
    \label{eq:SINR}
\end{equation}
where $\sigma^2$ is the receiver noise variance. Here, $\beta_{l,k}$ and $\gamma_{l,k}$ represent the large-scale fading coefficient for the channel from AP $l$ to UE $k$ and the mean-square of the corresponding channel estimate, respectively. $\delta_{l,k}$ denotes the membership decision of UE $k$, where 
$\delta_{l,k} = 1$, if  $k \in \mathcal{S}_l$, and $\delta_{l,k} =  0$, if  $k \in \mathcal{W}_l$. $\tau_{\mathcal{S}_l} \leq \tau_{p}$ denotes the number of pilot signals for the strong-channel UEs at AP $l$.

\subsection{Fronthaul Rate Requirement with Functional Splits}
Cell-free massive MIMO relies on coherent joint transmission, which can only be realized by low-PHY functional split options \cite{larsen_survey_2019}.  Although higher split options allow for a lower fronthaul rate, they limit the phase synchronization between APs and prohibit coherent joint transmission. In this work, we focused on options $8$ and  $7.1$, both scaling the fronthaul load with the active number of antennas. Based on the chosen functional split, the required fronthaul rate for an AP changes. For option $8$, all IQ data need to be transmitted over fronthaul, where the rate requirement becomes 
\begin{equation}
    \bar{R}_{8} = 2 \Delta f N_{\mathrm{FFT}} N_{\mathrm{bits}} M,
\end{equation} where $M$ is the number of active antennas at the AP, $\Delta f$ is the subcarrier spacing, $N_{\mathrm{FFT}}$ is the fast Fourier transform (FFT) size, and $N_{\mathrm{bits}}$ is the number of quantization bits.
In the split option $7.1$, the filtering and FFT are realized at the AP-site, lowering the fronthaul rate requirement only with active data subcarriers:
\begin{equation}
    \bar{R}_{7.1} = 2 \Delta f N_{\mathrm{used}} N_{\mathrm{bits}} M,
\end{equation}
where $N_{\mathrm{used}} = \lfloor B^{\mathrm{ac}} / \Delta f \rfloor $ is the active data subcarriers. Since $N_{\mathrm{used}} \leq N_{\mathrm{FFT}}$, we can guarantee $\bar{R}_{7.1} \leq \bar{R}_{8}$. In both cases, the fronthaul rate scales with the number of active antenna elements at the AP, which will be observed to create a significant performance bottleneck in the numerical analysis.

\subsection{Wireless Fronthaul Signal Transmission}
\label{sec:wireless_fronth}
We consider space-division multiple access (SDMA) for the fronthaul channels.  The cloud unit is equipped with $M_c$ antennas driven by  $N_c$ RF chains with $N_c \ll M_c$ and $N_c \geq L$. The received fronthaul signal at AP $l$ is 
\vspace{-2mm}
\begin{equation}
    \vect{y}_{l} = \vect{G}_{l} \vect{F} \vect{W} \vect{s} + \vect{n}_{l}, \vspace{-2mm}
\end{equation}
where $ \vect{G}_{l} \in \mathbb{C}^{   M^{\mathrm{frh}} \times M_c }$ is the downlink fronthaul channel between the cloud and AP $l$, and $\vect{n}_{l}\sim \mathcal{N}_{\mathbb{C}}(\vect{0}, \sigma^2 \mathbf{I}_{M^{\mathrm{frh}}})$ is the additive white Gaussian noise (AWGN). Since the fronthaul links are assumed to be in LOS and the AP deployments are static, the fronthaul channel is perfectly known to the cloud.  $\vect{s}=[s_1,\ldots, s_L]^{\Ttran} \in \mathbb{C}$ denotes the downlink message signal of all APs. $\mathbb{E}[|s_l|^2] = \bar{p}_l$, where $\bar{p}_l$ denotes the fronthaul transmit signal power for the $l$th AP's message signal, where $\sum_{l=1}^{L} \bar{p}_l = P_f$, and $P_f$ is the fronthaul transmit power budget. 
$\vect{F} \in \mathbb{C}^{M_c \times N_c}$ and $\vect{W}\in \mathbb{C}^{N_c \times L}$ are the analog and digital precoding matrices. The APs also perform analog combining with the combining vector $\vect{v}_l \in \mathbb{C}^{M^{\mathrm{frh}} }$, which are chosen based on the angles of arrival from the cloud to the corresponding AP.   

We assume that the cloud selects the columns of $\vect{F}$ as array response vectors matched to the directions of the corresponding APs. By including the effects of analog beamforming into the channel, we can characterize equivalent channel representation as $(\vect{g}^{\mathrm{eq}}_l)^{\Htran} = \vect{v}_l^{\Htran} \vect{G}_{l} \vect{F} $, and $\vect{\Bar{G}} = [\vect{g}^{\mathrm{eq}}_1, \ldots, \vect{g}^{\mathrm{eq}}_{L}]^{\Htran}$. Applying ZF precoding at the cloud, we obtain the achievable data rate of AP $l$ as \cite[Ch. 6]{bjornson2024introduction}
\begin{equation}
R^f_l = B^{\mathrm{frh}}\log_2 \left(1+ G_{l} \bar{p}_l \right),
\end{equation}
where $G_{l}$ is equal to $1/(\sigma^2\left[(\vect{\Bar{G}}\vect{\Bar{G}}^{\Htran})^{-1}\right]_{l,l})$.  

To provide connectivity through AP $l$, the achievable data rate at the fronthaul should satisfy 
\begin{equation}
 \bar{R}_{\mathcal{X}} =   O_{\mathcal{X}} M_l \leq R^f_l,
\end{equation}
where $\mathcal{X} \in \{8, 7.1 \}$, $O_{8} = 2 \Delta f N_{\mathrm{FFT}} N_{\mathrm{bits}}$ and $O_{7.1} = 2 \Delta f N_{\mathrm{used}} N_{\mathrm{bits}}$.  This demonstrates that the number of active antennas that an AP can use in the access link is limited by its fronthaul data rate. 
        \vspace{-2mm}
\section{Joint Access and Fronthaul Resource Allocation}
 \vspace{-2mm}
The allocated transmit powers for APs on the fronthaul and the activated access antennas are connected through the wireless fronthaul load constraints. Therefore, we propose joint fronthaul power and access antenna allocation algorithms in this section. We first consider the single-UE case to simplify the problem structure and build intuition for the antenna activation rule under wireless fronthaul limitations. Later, we propose a max-min fair resource allocation algorithm for multiple UEs.
        \vspace{-1mm}
\subsection{Single-UE Case}
        \vspace{-1mm}

In this special case, we assume a single UE is served in a given time/frequency block, with perfect channel estimation. The UE index $k$ is removed from the definitions in Section \ref{sec:system_model}, and  the received signal at the UE can be expressed as 
\vspace{-1mm}
\begin{equation}
    r = \sum_{l=1}^{L} \sqrt{\rho_l} \vect{h}^{\Ttran}_{l}\vect{w}_{l} \varsigma + n, \vspace{-1mm}
\end{equation}
where $n \sim \mathcal{N}_{\mathbb{C}}(0, \sigma^2)$ is the AWGN. The optimal precoding for this case is the distributed maximum ratio transmission (MRT), where $\vect{w}_{l} = \vect{h}^{\dag}_{l} / \|   \vect{h}_{l} \|$ and $\rho_l=P_t, \forall l$. Here, $P_t$ is the per-AP power limit. An upper bound on the UE's ergodic data rate is obtained by Jensen's inequality:
\vspace{-1mm}
\begin{equation}
    R \leq B^{\mathrm{ac}} \log_2\left(1+   \gamma \mathbb{E}\left[\left(\sum_{l=1}^{L} \|\vect{h}_{l}\| \right)^2\right]  \right),  \vspace{-1mm}
\end{equation}
where $\gamma = P_t/\sigma^2$. Using large-array approximation, the effective downlink SNR of the UE can be approximated as\footnote{When the channels are pure LOS, the SNR given \eqref{eq:snr} is the exact one. }
\begin{equation} \vspace{-1mm}
    \mathrm{SNR} \approx  \gamma \left(\sum_{l=1}^{L} \sqrt{\beta_{l}  M_l } \right)^2,\label{eq:snr} \vspace{-1mm}
\end{equation}
where $R \approx B^{\mathrm{ac}} \log_2(1+ \mathrm{SNR})$. 
The active number of antennas is an integer variable that determines the effective SNR and is limited by the fronthaul. Defining them as integer variables would make the optimization problem non-convex and combinatorial. To obtain an efficient solution, we first relax them into a continuous variable: $M_l \in [0, M^{\mathrm{ac}}]$. After this relaxation, we can derive closed-form conditions and a water-filling–type algorithm that activates additional antennas based on their relative proximity to the UE and the cloud. 
We define  $\hat{M}_l = \sqrt{M_l}, \hat{\beta}_l =  \sqrt{\gamma\beta_l}, \, \forall l$. The maximum downlink rate under wireless fronthaul limitation is obtained by solving the following problem:
\vspace{-1mm}
\begin{subequations} \label{eq:maxmin_opt1_singleUE:problem}
\begin{align}
 & \underset{\{\hat{\vect{M}}, \bar{\vect{p}} \}}{\text{maximize}} \quad \sum_{l=1}^{L}  \hat{\beta}_l  \hat{M}_l \label{eq:maxmin_opt1_singleUE:objective} \\ & \textrm{subject to} \nonumber \\  &  \hat{M}^2_l  Q_{\mathcal{X}} \leq   \log_2\left( 1 + {G_{l} \bar{p}_l}  \right)  , \quad \forall l ,\label{eq:maxmin_opt1_singleUE:wireless_fronthaul_rate} \\ &
\sum_{l = 1}^{L} \bar{p}_l \leq P_f ,\label{eq:maxmin_opt1_singleUE:wireless_fronthaul_power} 
\end{align}
\end{subequations}
where $\hat{\vect{M}}\! \!=\!\! [\hat{M}_1, \ldots, \hat{M}_L]^{\Ttran}$ and $Q_{\mathcal{X}} = O_{\mathcal{X}} / B^{\mathrm{frh}}$. Note that we replaced $\left(\sum_{l=1}^{L}  \hat{\beta}_l  \hat{M}_l \right)^2$ in effective SNR with $\left(\sum_{l=1}^{L}  \hat{\beta}_l  \hat{M}_l \right)$ in the objective. Since $\hat{\beta}_l$ and $\hat{M}_l$ are nonnegative real numbers, these two objectives are equivalent to one another. 
\vspace{-2mm}
\begin{lemma}
    At the optimal solution of \eqref{eq:maxmin_opt1_singleUE:problem}, the constraint \eqref{eq:maxmin_opt1_singleUE:wireless_fronthaul_rate} holds with equality, i.e., 
    $\hat{M}^*_l  Q_{\mathcal{X}} = \sqrt{\log_2\left( 1 + {G_{l} \bar{p}^*_l}  \right)}, \, \forall l$. 
    \label{lemma1}
\end{lemma}
\vspace{-3mm}
\begin{proof}
  The feasible set of \eqref{eq:maxmin_opt1_singleUE:problem} is nonempty,
    and the objective is strictly increasing in each $\hat{M}_l$.
    Hence, at optimality, every constraint
    \eqref{eq:maxmin_opt1_singleUE:wireless_fronthaul_rate} must be active;
    otherwise, $\hat{M}_l$ could be increased, contradicting optimality.
    \vspace{-2mm}
\end{proof}

Using Lemma \ref{lemma1}, we can rewrite  \eqref{eq:maxmin_opt1_singleUE:problem} only in terms of $\bar{\vect{p}}$: 
\begin{subequations} \label{eq:maxmin_opt2_singleUE:problem}
\begin{align}
 & \underset{\{ \bar{\vect{p}} \}}{\text{maximize}} \quad \sum_{l=1}^{L}  \hat{\beta}_l q_{\mathcal{X}}  \sqrt{\log_2\left( 1 + {G_{l} \bar{p}_l}  \right)} \label{eq:maxmin_opt2_singleUE:objective} \\ & \textrm{subject to} \quad \sum_{l = 1}^{L} \bar{p}_l \leq P_f,\label{eq:maxmin_opt2_singleUE:wireless_fronthaul_power} 
\end{align}
\end{subequations}
where  $q_{\mathcal{X}} = Q^{-1}_{\mathcal{X}}$. The following lemma provides the optimal solution of \eqref{eq:maxmin_opt2_singleUE:problem}.

\vspace{-1mm}
\begin{lemma}
    The optimal power allocation $\bar{p}^*_l$ for problem
    \eqref{eq:maxmin_opt2_singleUE:problem} is 
   \begin{equation}
        \bar{p}^*_l = \max\left\{0,\,
            \frac{\sqrt{2}\,\hat{A}_l}
                 {\lambda^*\sqrt{W\!\left(2G^2_l\hat{A}_l^2/{\lambda^*}^2\right)}}
            - \frac{1}{G_l}
        \right\}, \quad \forall l,
        \label{eq:optimal_power}
    \end{equation}
    where
    \begin{equation}
        \hat{A}_l = \frac{\sqrt{q_{\mathcal{X}}\ln 2}\; \,\hat{\beta}_l}{2}, \quad \forall l.
    \end{equation}
    $W(\cdot)$ is the standard form of Lambert-W function, and $\lambda^*$ is the optimal
    dual variable obtained via bisection, such that
    $\sum_{l=1}^{L} \bar{p}^*_l = P_f$.
    \label{lemma:water_filling}
\end{lemma}
\vspace{-3mm}
\begin{proof}
The proof is given in the journal version of this paper. It follows from taking the derivative of the Lagrangian function and obtaining the KKT conditions. 
\vspace{-2mm}
\end{proof}
      \vspace{-1mm}
The solution in Lemma \ref{lemma:water_filling} demonstrates a water-filling solution, where the water level is scaled by a weighted Lambert-W function.

      \vspace{-1mm}
\subsection{Multi-UE Case}
        \vspace{-1mm}
\label{sec:ratemax_wireless}
In this section, we will introduce an optimization problem that maximizes the minimum UE SINR under wireless fronthaul limitations. This will allow us to analyze how much wireless fronthaul is limiting the performance. Under split options 7.1 and 8, this problem can be given by

\begin{subequations} \label{eq:WFmaxmin0:problem}
\begin{align}
 & \underset{\{\vect{M}, \boldsymbol{\rho}, \bar{\vect{p}}, \boldsymbol{\alpha}, \upsilon \}}{\text{maximize}} \quad  \upsilon \label{eq:WFmaxmin0:objective} \\
     & \textrm{subject to} \nonumber \\ &\frac{\left(\sum_{l=1}^L \sqrt{\left(M_l-\tau_{\mathcal{S}_l}\right) \alpha_l \rho_{l, k} \gamma_{l, k}}\right)^2}{\sum_{t=1}^K \sum_{l=1}^L \rho_{l, t}\left(\beta_{l, k}-\delta_{l, k} \gamma_{l, k}\right)+\sigma^2}\geq \upsilon, \forall k ,\label{eq:WFmaxmin0:SINR_constraint} \\
 & Q_{\mathcal{X}} M_l \alpha_l \leq   \log_2\left( 1 + {G_{l} \bar{p}_l}  \right)    , \quad \forall l ,\label{eq:WFmaxmin0:wireless_fronthaul_rate} 
 \\&
\sum_{l = 1}^{L} \bar{p}_l \leq P_f, \quad  \forall i  ,\label{eq:WFmaxmin0:wireless_fronthaul_power} 
\\& 
 \sum_{k=1}^{K} \rho_{l,k} \leq  P_t  \alpha_l, \quad \forall l,\label{eq:WFmaxmin0:access_power_limit} \\&
M_l \in \{0, \tau_{\mathcal{S}_l},\ldots,M^{\mathrm{ac}}\}, \quad \forall l, \label{eq:WFmaxmin0:integer}  \\&
\alpha_l \in  \{ 0, 1 \} , \quad \forall l ,\label{eq:WFmaxmin0:binary}
\end{align}
\end{subequations}
where $\boldsymbol{\rho} \in \mathbb{R}_{+}^{L\times K}$, and $\rho_{l,k}$ is the $l$th row and $k$th column element of $\boldsymbol{\rho}$. $\boldsymbol{\alpha}=[\alpha_1, \ldots, \alpha_L]^{\Ttran}$ is the AP activation vector, where $\alpha_l$ denotes the activation of AP $l$, ensuring that if AP cannot afford activating more than $\tau_{\mathcal{S}_l}$ antennas, \eqref{eq:WFmaxmin0:wireless_fronthaul_rate} is satisfied by deactivating the radio.   The objective function, \eqref{eq:WFmaxmin0:objective}, maximizes the minimum effective SINR of the UEs. \eqref{eq:WFmaxmin0:SINR_constraint} ensures the effective SINR at UE $k$ is higher than or equal to the threshold  $\upsilon$. \eqref{eq:WFmaxmin0:wireless_fronthaul_rate} ensures that the fronthaul rate for AP $l$ is higher than or equal to the required fronthaul rate.  \eqref{eq:WFmaxmin0:wireless_fronthaul_power} and \eqref{eq:WFmaxmin0:access_power_limit} limit the transmit power in the fronthaul and access links, respectively. \eqref{eq:WFmaxmin0:integer} ensures that the number of active antennas at AP $l$, $M_l$, is an integer variable either between $\tau_{\mathcal{S}_l}$ and the deployed number of antennas at AP $l$ or equal to zero, denoting deactivation. 
\eqref{eq:WFmaxmin0:problem} is a non-convex problem with a combinatorial nature due to the integer variables and the non-convex structure of the effective SINR.

As in the single-UE case, we start by relaxing the integer constraints on the active number of antennas by replacing \eqref{eq:WFmaxmin0:integer} with $0 \leq \tilde{M}_l   \leq  M^{\mathrm{ac}} - \tau_{\mathcal{S}_l} \in \mathbb{R}_+$. We will later recover $M_l$ by  rounding $\tilde{M}_l + \tau_{\mathcal{S}_l}$.

\begin{lemma}
   After using $\tilde{M}_l$, an efficient solution for  \eqref{eq:WFmaxmin0:problem} can be obtained by iteratively solving the following convex problem.
   \begin{subequations} \label{eq:WFmaxmin_opt1:problem}
\begin{align}
 & \underset{\{\tilde{\vect{M}}, \bar{\boldsymbol{\rho}}, \vect{z}, \vect{u}, \vect{v}, \vect{s}, \bar{\vect{p}}, \upsilon \}}{\text{maximize}} \quad \upsilon  - \Omega \sum^{L}_{l=1} s_l \label{eq:maxmin_opt1:objective} \\ & \quad \,\, \textrm{subject to} \quad  \eqref{eq:WFmaxmin0:wireless_fronthaul_power}, \eqref{eq:WFmaxmin0:access_power_limit},  \nonumber \\ &
\bar{\vect{b}}^{\Ttran}_k {\bar{\boldsymbol{\rho}}} +1 \leq  
\nabla {f^{(c)}_k}^{\Ttran}  \begin{bmatrix} \vect{z}^{(c+1)}_k -  \vect{z}^{(c)}_k \\ \upsilon^{(c+1)}  -  \upsilon^{(c)} \end{bmatrix}   +  {f^{(c)}_k},  \forall k, \\ &
\| \tilde{M}_l, \sqrt{2} v_l, 1+G_l \bar{p}_l \| \leq 1+G_l \bar{p}_l + \tilde{M}_l, \quad \forall l, \\ &
      u_l \leq 2v^{(c)}_l v^{(c+1)}_l - {v_l^{(c)}}^2 + s_l, \quad \forall l, \\ &
          Q'_{\mathcal{X}}\tilde{M}^2_l + Q'_{\mathcal{X}}\tilde{M}_l\tau_{\mathcal{S}_l} + \leq u_l \ln \left({u_l}/{\tilde{M}_l} \right), \quad \forall l, 
    \label{eq:CCP-FPP} 
\end{align}
\end{subequations}
where $\vect{z}_k = [ z_{1,k}, \ldots, z_{L,k}]^{\Ttran}$,  $ \boldsymbol{\bar{\gamma}}_k = [ \sqrt{\gamma_{1,k}}, \ldots, \sqrt{\gamma_{L,k}}]^{\Ttran}$, $\bar{\boldsymbol{\rho}} = [\sum_k \rho_{1k}, \ldots, \sum_k \rho_{Lk} ]^{\Ttran} $, $\bar{\vect{b}}_k = [\bar{b}_{1k}, \ldots, \bar{b}_{Lk}]^{\Ttran}$, $Q'_{\mathcal{X}} = Q_{\mathcal{X}} \ln(2)$,  $\bar{b}_{lk} = \left(\beta_{l, k}-\delta_{l, k} \gamma_{l, k}\right)$. $\Omega$ is the penalty coefficient to push feasibility parameters to zero. $ f_k= f(\vect{z}_k, \upsilon)  = \frac{  | \boldsymbol{\bar{\gamma}}_k^{\Ttran} \vect{z}_k |^2}{\upsilon } $, and $ \nabla {f^{(c)}_k} = \left[  \frac{ 2  \boldsymbol{\bar{\gamma}}_k \boldsymbol{\bar{\gamma}}_k^{\Ttran} \vect{z}^{(c)}_k     }{ \upsilon^{(c)} },  \frac{-  |\boldsymbol{\bar{\gamma}}_k^{\Ttran} \vect{z}^{(c)}_k|^2   }{ {\upsilon^{(c)}}^2 }  \right]^{\Ttran}$ is the gradient of $f^{(c)}_k$, where $c$ denotes the iteration index that will be used later. $s_l \geq 0, \, \forall l$ are slack variables, guaranteeing \eqref{eq:CCP-FPP} feasibility for any starting point. $\vect{u} = [u_1, \ldots, u_L]^{\Ttran} \in \mathbb{R}_{+}^L$ and $\vect{v} = [v_1, \ldots, v_L]^{\Ttran} \in \mathbb{R}_{+}^L$, are auxiliary  variables.
\end{lemma}
\begin{proof}
The proof is given in the journal version of this paper.  The proof relies on multiplying both sides of \eqref{eq:WFmaxmin0:wireless_fronthaul_rate} with $\tilde{M}_l$ and using successive convex approximation (SCA) and feasible point pursuit (FPP) to guarantee feasibility.    
\end{proof}
After solving \eqref{eq:WFmaxmin_opt1:problem} until convergence, the obtained antenna numbers are rounded as in \cite{energyjournal}. 

        \vspace{-3mm}
\section{Numerical Analysis}
        \vspace{-1mm}

We consider a square area of size  $1~\text{km}^2$ with a grid-type AP deployment. The cloud is located at $[0,480]$\,m, at the edge of the considered area to investigate unfairness among APs due to the wireless fronthaul limitation. In the fronthaul, we consider the cloud is equipped with a $64\times 4$ uniform planar array (UPA), and the APs are equipped with uniform linear arrays (ULAs).  We consider $L=25$, and  $M^{\rm ac} = 8$. We consider $5$\,GHz and $28$\,GHz carrier frequency for the access and fronthaul links, respectively. While the fronthaul channel is LOS dominant, uncorrelated Rayleigh fading is assumed in the access channel. The shadowing effect in the access channel is modeled as in \cite{ozlem_jsac}. The UEs are distributed uniformly in the considered area. We run $150$ Monte Carlo simulations for the single-UE case simulations, and $50$ Monte Carlo simulations for the multi-UE case simulations. We used CVX \cite{cvx} in solving \eqref{eq:maxmin_opt1_singleUE:problem} and \eqref{eq:WFmaxmin_opt1:problem}. We set $\Omega = 10^{-3}$, where convex-concave programming and feasible point pursuit algorithm iteratively solves \eqref{eq:WFmaxmin_opt1:problem} until normalized mean-squared error (NMSE) of $\upsilon^{(c)}$ in subsequent iterations is below $10^{-2}$.

\begin{table}[tb!]
\caption{Simulation parameters}
\centering
\begin{tabular}{|l|l|l|l|}
\hline
$M^{\rm frh}$,  $M_{c}$ & 64, 256 & $N_{\mathrm{bits}}$ & 12 \\ \hline
 $B^{\rm frh}$ &  $1000$\,MHz   &  $\Delta f$ & $30$\,kHz \\ \hline
 $P_f$  & $25$  &  $\tau_c$, $\tau_p$  & $260$, $6$   \\ \hline
\end{tabular}
\label{tab:simulation_params}
\vspace{-2mm}
\end{table}

\begin{figure}[tb]
        \vspace{-3mm}
    \centering
    \includegraphics[width=0.85\linewidth]{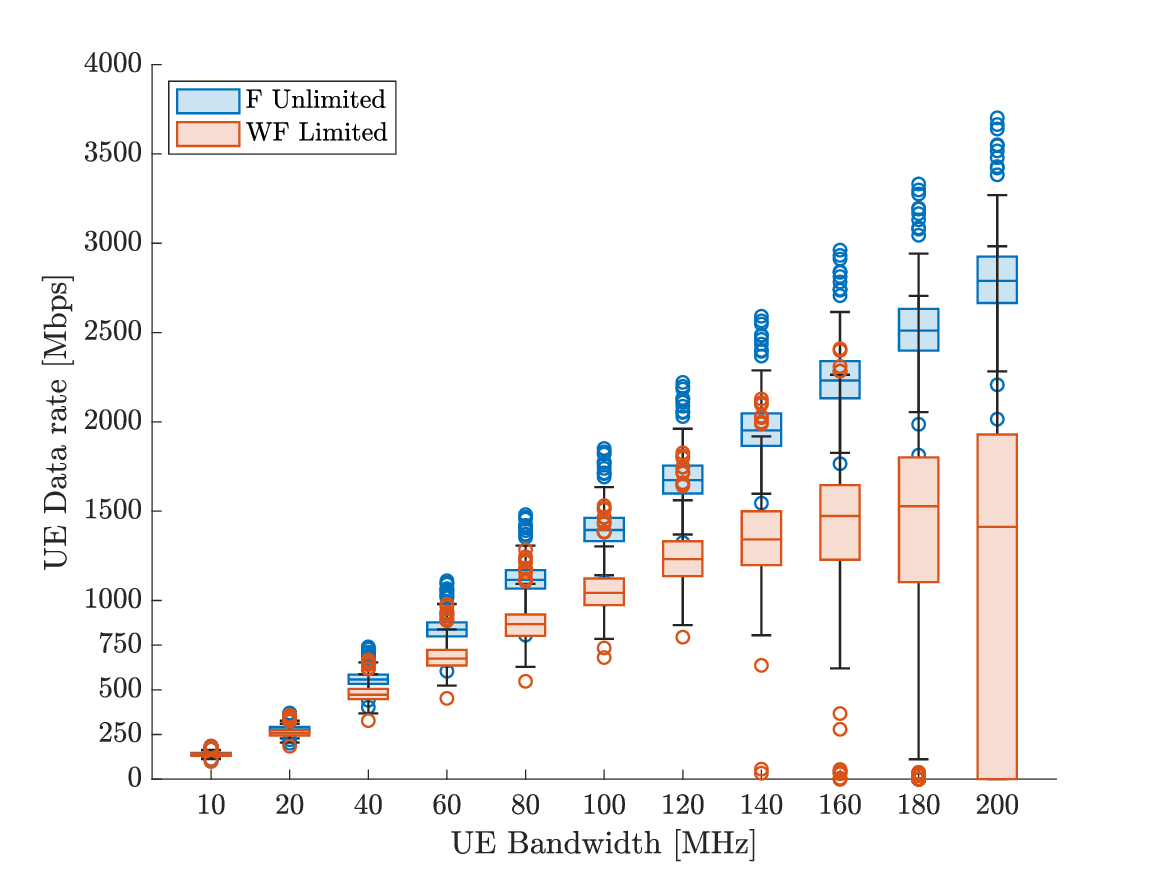}
        \vspace{-4mm}
    \caption{The box plot of UE data rate vs UE bandwidth for unlimited fronthaul and wireless fronthaul (WF) limitation cases.}
    \label{fig:singleUEBandwdith}
\end{figure}

Fig. \ref{fig:singleUEBandwdith} illustrates the UE data rate with a varying UE bandwidth, considering unlimited fronthaul and wireless fronthaul limitation (under split option $7.1$) in this single-UE case. The increments in UE bandwidth are implemented by increasing $N_{\rm used}$. $B^{\rm frh}$ is set to $200$\,MHz for this figure. At lower UE bandwidths, the wireless fronthaul limitation is not observed, demonstrating that almost all antennas at APs can be activated. As the UE bandwidth increases, the fronthaul constraints become harder to meet, resulting in fewer antenna activations. As shown in the figure, the overall UE data rate not only destabilizes but also decreases due to significant radio deactivation. Furthermore, the wireless fronthaul limitation also creates inequality based on the UE's location. The UEs closer to the cloud are less affected by wireless fronthaul limitations than those farther away. 

\begin{figure}[tb]
    \subfigure[]{
        \includegraphics[width=0.49\linewidth]{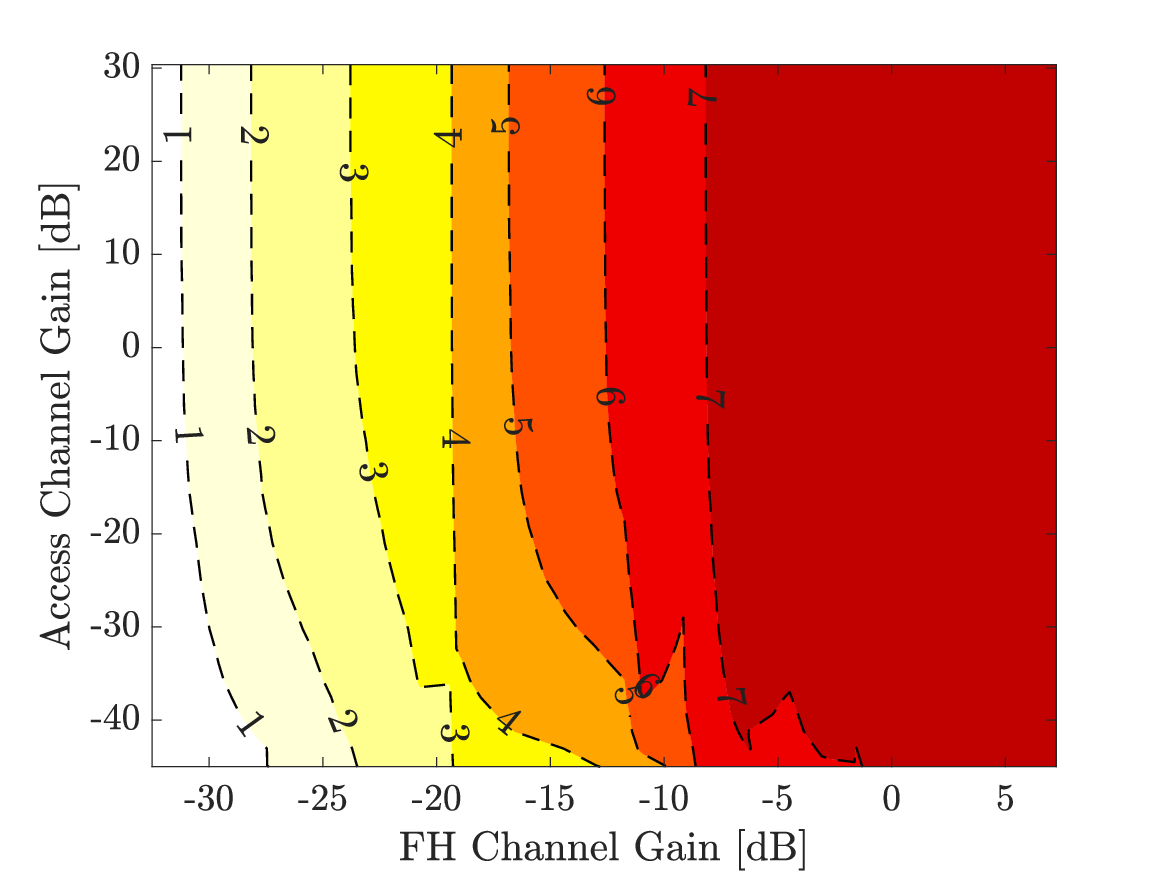} }
        \hspace{-0.5cm}
   \subfigure[]{
         \includegraphics[width=0.49\linewidth]{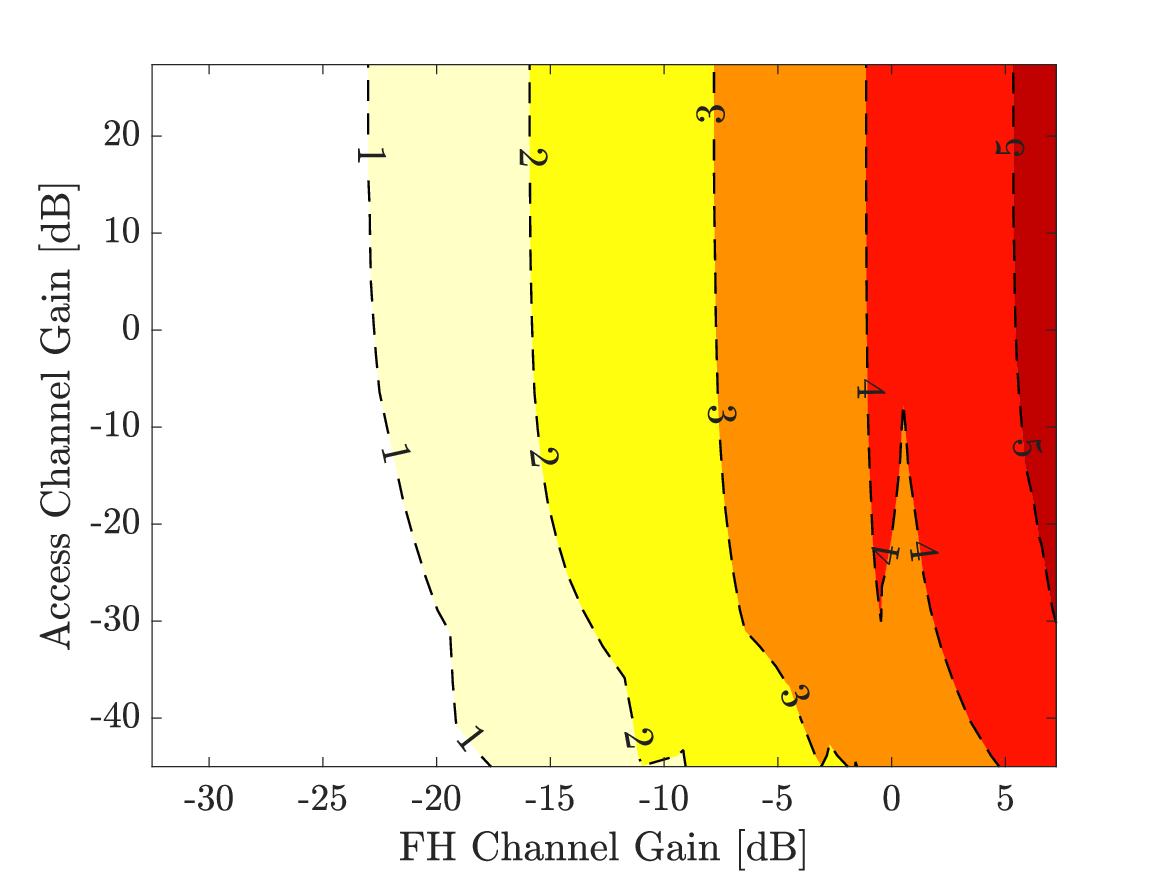}  }
                 \vspace{-5mm}
    \caption{Number of active antennas for varying fronthaul and access channel gains for UE bandwidth is equal to (a) 20\,MHz, (b)  40\,MHz.}
    \label{fig:AntennaActivation}
\end{figure}

Fig. \ref{fig:AntennaActivation} shows the active number of antennas under different AP-cloud and AP-UE channel gains. When the UE bandwidth is $20$\,MHz, the fronthaul limitation is less stringent, allowing higher antenna activation. Interestingly, the antenna activation is mainly affected by the fronthaul channel gain rather than the access channel gain, regardless of the UE bandwidth. Under severe access channel fading cases, fewer antennas are activated in the radios, demonstrating a resource saving for the APs having better channel gain with the UE. 

\begin{figure}[tb]
        \vspace{-5mm}
    \centering
    \includegraphics[width=0.82\linewidth]{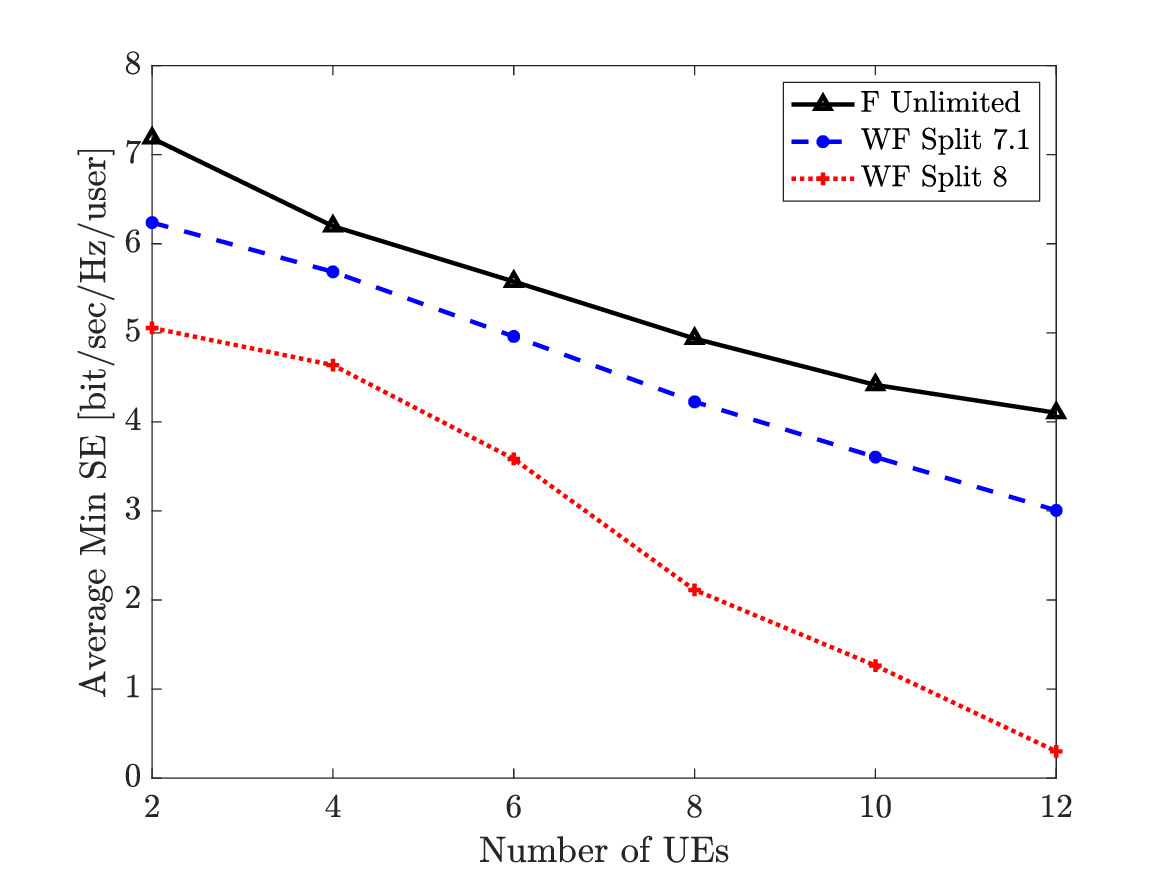}
           \vspace{-4mm}
    \caption{Average of minimum UE SE for varying number of UEs. F Unlimited, and WF indicate unlimited fronthaul and wireless fronthaul cases, respectively.}
    \label{fig:multiUEK}
    \vspace{-4mm}
\end{figure}

Fig. \ref{fig:multiUEK} shows the average minimum SE versus the number of UEs under different fronthaul configurations. The unlimited fronthaul (F unlimited) benchmark is obtained by the max-min-fairness algorithm proposed in \cite{interdonato2020local}, in which only the power allocation is optimized, assuming all antennas at APs are active. In contrast, WF schemes exhibit progressively lower SE as the number of UEs increases, reflecting the growing fronthaul bottleneck as more antennas are required to be activated to cancel out interference. Split 7.1 achieves moderate performance by balancing processing and fronthaul usage, while Split 8, which imposes a higher fronthaul burden, leads to severe performance degradation.

        \vspace{-1mm}
\section{Conclusion}
        \vspace{-1mm}

In this work, we analyzed the performance of a cell-free massive MIMO network with wireless fronthaul. Under realistic low-PHY functional splits, wireless fronthaul constraints the number of active antennas at the APs. We proposed joint access and fronthaul resource allocation algorithms that maximize the minimum UE spectral efficiency while ensuring fronthaul load constraints are satisfied. Our analysis shows that power allocation over the wireless fronthaul follows a modified water-filling scheme, where the water level depends jointly on the access and fronthaul channel gains. Severe fronthaul limitations not only reduce UE rates but also introduce spatial performance disparities related to the cloud location. 
           \vspace{-2mm}
\bibliographystyle{IEEEtran}
\bibliography{IEEEabrv,refs}

\begin{thebibliography}{10}
\providecommand{\url}[1]{#1}
\csname url@samestyle\endcsname
\providecommand{\newblock}{\relax}
\providecommand{\bibinfo}[2]{#2}
\providecommand{\BIBentrySTDinterwordspacing}{\spaceskip=0pt\relax}
\providecommand{\BIBentryALTinterwordstretchfactor}{4}
\providecommand{\BIBentryALTinterwordspacing}{\spaceskip=\fontdimen2\font plus
\BIBentryALTinterwordstretchfactor\fontdimen3\font minus \fontdimen4\font\relax}
\providecommand{\BIBforeignlanguage}[2]{{%
\expandafter\ifx\csname l@#1\endcsname\relax
\typeout{** WARNING: IEEEtran.bst: No hyphenation pattern has been}%
\typeout{** loaded for the language `#1'. Using the pattern for}%
\typeout{** the default language instead.}%
\else
\language=\csname l@#1\endcsname
\fi
#2}}
\providecommand{\BIBdecl}{\relax}
\BIBdecl

\bibitem{cfmMIMOOr}
H.~Q. Ngo, A.~Ashikhmin, H.~Yang, E.~G. Larsson, and T.~L. Marzetta, ``Cell-free massive {MIMO} versus small cells,'' \emph{IEEE Transactions on Wireless Communications}, vol.~16, no.~3, pp. 1834--1850, 2017.

\bibitem{larsen_survey_2019}
L.~M.~P. Larsen, A.~Checko, and H.~L. Christiansen, ``A {Survey} of the {Functional} {Splits} {Proposed} for {5G} {Mobile} {Crosshaul} {Networks},'' \emph{IEEE Communications Surveys \& Tutorials}, vol.~21, no.~1, pp. 146--172, 2019, conference Name: IEEE Communications Surveys \& Tutorials.

\bibitem{deployment_IDS}
O.~A. Topal, {\"O}.~T. Demir, E.~Björnson, and C.~Cavdar, ``A novel access point deployment framework for {mmWave} cell-free massive {MIMO} networks,'' \emph{IEEE Transactions on Wireless Communications}, vol.~24, no.~6, pp. 4581--4597, 2025.

\bibitem{cell-free-book}
\BIBentryALTinterwordspacing
{\"O}.~T. Demir, E.~Bj\"{o}rnson, and L.~Sanguinetti, ``Foundations of user-centric cell-free massive {MIMO},'' \emph{Foundations and Trends® in Signal Processing}, vol.~14, no. 3-4, pp. 162--472, 2021. [Online]. Available: \url{http://dx.doi.org/10.1561/2000000109}
\BIBentrySTDinterwordspacing

\bibitem{umurhan}
U.~Demirhan and A.~Alkhateeb, ``Enabling cell-free massive {MIMO} systems with wireless millimeter wave fronthaul,'' \emph{IEEE Transactions on Wireless Communications}, vol.~21, no.~11, pp. 9482--9496, 2022.

\bibitem{wireless_f}
S.~Elhoushy, M.~Ibrahim, and W.~Hamouda, ``Downlink performance of {CF} massive {MIMO} under wireless-based fronthaul network,'' \emph{IEEE Transactions on Communications}, vol.~71, no.~5, pp. 2632--2653, 2023.

\bibitem{energyjournal}
O.~A. Topal, {\"O}.~T. Demir, E.~Björnson, and C.~Cavdar, ``Unlocking the energy-saving potential in {O-RAN} cell-free massive {MIMO} by joint orchestration of radio, wireless fronthaul, and cloud resources,'' \emph{IEEE Transactions on Wireless Communications}, 2026, under review.

\bibitem{icton}
O.~A. Topal, {\"O}.~T. Demir, and C.~Cavdar, ``Rethinking energy efficiency in cell-free massive {MIMO}: The role of processing and optical fronthaul,'' \emph{arXiv preprint arXiv:2606.31412}, 2026.

\bibitem{AsilomarConf24}
O.~A. Topal, O.~T. Demir, E.~Bjornson, and C.~Cavdar, ``Energy-efficient cell-free massive {MIMO} with wireless fronthaul,'' in \emph{2024 58th Asilomar Conference on Signals, Systems, and Computers}, 2024, pp. 1591--1596.

\bibitem{neetu}
N.~R.R., O.~A. Topal, {\"O}.~T. Demir, E.~Björnson, C.~Cavdar, G.~Ghatak, and V.~A. Bohara, ``Uav-based cell-free massive {MIMO}: {J}oint activation and power optimization under fronthaul capacity limitations,'' \emph{IEEE Wireless Communications Letters}, pp. 1--1, 2025.

\bibitem{interdonato2020local}
G.~Interdonato, M.~Karlsson, E.~Bj{\"o}rnson, and E.~G. Larsson, ``Local partial zero-forcing precoding for cell-free massive {MIMO},'' \emph{IEEE Transactions on Wireless Communications}, vol.~19, no.~7, pp. 4758--4774, 2020.

\bibitem{bjornson2024introduction}
E.~Bj{\"o}rnson and {\"O}.~T. Demir, ``Introduction to multiple antenna communications and reconfigurable surfaces,'' \emph{Now Publishers, Inc.}, 2024.

\bibitem{ozlem_jsac}
{\"O}.~T. Demir, M.~Masoudi, E.~Björnson, and C.~Cavdar, ``Cell-free massive {MIMO} in {O-RAN}: Energy-aware joint orchestration of cloud, fronthaul, and radio resources,'' \emph{IEEE Journal on Selected Areas in Communications}, vol.~42, no.~2, pp. 356--372, 2024.

\bibitem{cvx}
M.~Grant and S.~Boyd, ``{CVX}: Matlab software for disciplined convex programming,'' \url{http://cvxr.com/cvx}, Apr. 2011.

\end{thebibliography}

\end{document}